\documentclass[11pt, a4paper]{article}
\usepackage{amsmath, amsthm, amssymb, url}
\usepackage[margin=1.2in]{geometry}
\usepackage[english]{babel}
\usepackage{multirow}
\usepackage{authblk}
\usepackage{tikz-cd} 
\usepackage{color}
\usepackage{hyperref}
\usepackage[linesnumbered, ruled, vlined]{algorithm2e}
\newcommand{\Fix}{\mathrm{Fix}}

\newcommand{\MN}{\mathrm{MN}}

\theoremstyle{plain}

\newtheorem{corollary}{Corollary}
\newtheorem{lemma}{Lemma}
\newtheorem{proposition}{Proposition}

\newtheorem{theorem}{Theorem}

\theoremstyle{definition}
\newtheorem{definition}{Definition}
\newtheorem{example}{Example}
\newtheorem{remark}{Remark}

\begin{document}

\title{Connections between the minimal neighborhood and the activity value of cellular automata}
\author[1]{Alonso Castillo-Ramirez\footnote{Email: alonso.castillor@academicos.udg.mx}}
\author[2]{Eduardo Veliz-Quintero \footnote{Email: eduardo.veliz9236@alumnos.udg.mx}}
\affil[1]{Centro Universitario de Ciencias Exactas e Ingenier\'ias, Universidad de Guadalajara, M\'exico.}
\affil[2]{Centro Universitario de los Valles, Universidad de Guadalajara, M\'exico.}

\maketitle

\begin{abstract}
For a group $G$ and a finite set $A$, a cellular automaton is a transformation of the configuration space $A^G$ defined via a finite neighborhood and a local map. Although neighborhoods are not unique, every CA admits a unique \emph{minimal neighborhood}, which consists on all the essential cells in $G$ that affect the behavior of the local map. An \emph{active transition} of a cellular automaton is a pattern that produces a change on the current state of a cell when the local map is applied. In this paper, we study the links between the minimal neighborhood and the number of active transitions, known as the \emph{activity value}, of cellular automata. Our main results state that the activity value usually imposes several restrictions on the size of the minimal neighborhood of local maps.   \\

\end{abstract}

\section{Introduction}

Cellular automata (CA) are mathematical models used to simulate complex systems over discrete spaces that have the key feature of being defined by a fixed local rule that is applied homogeneously and in parallel in the whole space. Traditionally, CA are defined only over a \emph{universe} that is a $d$-dimensional grid $\mathbb{Z}^d$ (e.g., see \cite{Kari}), but recent studies have considered the more general setting of a universe that is an arbitrary group $G$ (e.g., see \cite{CAandG}).         

Before defining CA more formally, we shall introduce the \emph{configuration space} $A^G$, which consists of the set of all functions, or \emph{configurations}, $x : G \to A$, where $A$ is a finite set known as the \emph{alphabet}. When $S$ is a finite subset of $G$, a function $p : S \to A$ is called a \emph{pattern} (or a \emph{block}) over $S$, and the set of all patterns over $S$ is denoted by $A^S$. 

\begin{definition}\label{def-ca}
A \emph{cellular automaton} is a transformation $\tau : A^G \to A^G$ such that there exists a finite subset $S \subseteq G$, called a \emph{neighborhood} of $\tau$, and a \emph{local map} $\mu : A^S \to A$ satisfying
\[ \tau(x)(g) = \mu( (g \cdot x) \vert_S), \quad \forall x \in A^G, g \in G,\]
where $(g \cdot x) \in A^G$ is the \emph{shift} of $x$ by $g$ defined by
\[ (g \cdot x)(h) := x(hg), \quad \forall h \in G.  \]
\end{definition}

Intuitively, applying $\tau$ to a configuration $x \in A^G$ is the same as applying the local map $\mu : A^S \to A$ homogeneously and in parallel using the shift action of $G$ on $A^G$. Local maps $\mu : A^S \to A$ are also known in the literature as \emph{block maps}. A well-known result is the \emph{Curtis-Hedlund-Lyndon theorem} which states that a function $\tau : A^G \to A^G$ is a cellular automaton if and only if $\tau$ is \emph{$G$-equivariant} (i.e. $\tau(g \cdot x) = g \cdot \tau(x)$, for all $g \in G$, $x \in A^G$) and continuous in the \emph{prodiscrete topology} of $A^G$ (i.e. the product topology of the discrete topology of $A$).  

Cellular automata do not have a unique neighborhood. With the above notation, for any finite superset $S^{\prime} \supseteq S$, we may define $\mu^\prime : A^{S^{\prime}} \to A$ by $\mu^\prime(z) := \mu(z \vert_S)$, for all $z \in A^{S^{\prime}}$, and it follows that $\mu^\prime$ is also a local map that defines $\tau$. Hence, any finite superset of a neighborhood of $\tau$ is also a neighborhood of $\tau$. However, cellular automata do have a unique \emph{minimal neighborhood} which consists of all the \emph{essential} elements of $G$ required to define a local defining map for $\tau$ (see Propositions \ref{prop-mn} and \ref{le-mms}). Minimal neighborhoods do not behave well with respect to the composition of cellular automata (see \cite[Ex. 1.27]{ExCA}), neither with respect to the inverse of an invertible cellular automaton. 

If the identity $e$ of the group $G$ is in a neighborhood $S \subseteq G$, a pattern $p \in A^S$ is called an \emph{active transition} of the local map $\mu : A^S \to A$ if
\[ \mu(p) \neq p(e).  \]
This notion has been recently and independently considered by several authors. In \cite{Pedro1, Pedro2, Fates1, Fates2}, active transitions of local maps are used to study the behavior of asynchronous one-dimensional cellular automata. The term \emph{activity value} $\alpha(\mu) \in \mathbb{N}$ of $\mu : A^S \to A$ is introduced in \cite{Concha} as the number of active transitions of $\mu$ in order to define a notion of \emph{sub-rule} among elementary cellular automata. 

Computational experiments showed that there is a clear connection between the possible sizes of the minimal neighborhood and the activity value of a local function. Figure \ref{fig1} shows this connection for local functions of the form $A^S \to A$ with $\vert A \vert = 2$, $\vert S \vert = 5$ and $S \subseteq \mathbb{Z}$. In this figure, the horizontal axis shows the possible activity values (the maximum activity is $2^5$), the vertical axis shows the possible sizes of the minimal neighborhood, and the existence of a vertical bar in the graphs shows the existence of a local function with the corresponding activity value and size of minimal neighborhood. Figure \ref{fig2}, is analogous to Figure \ref{fig1}, but for the case $\vert A \vert = 3$, $\vert S \vert = 3$ and $S \subseteq \mathbb{Z}$.

\begin{figure}
\centering
\includegraphics[scale=.7]{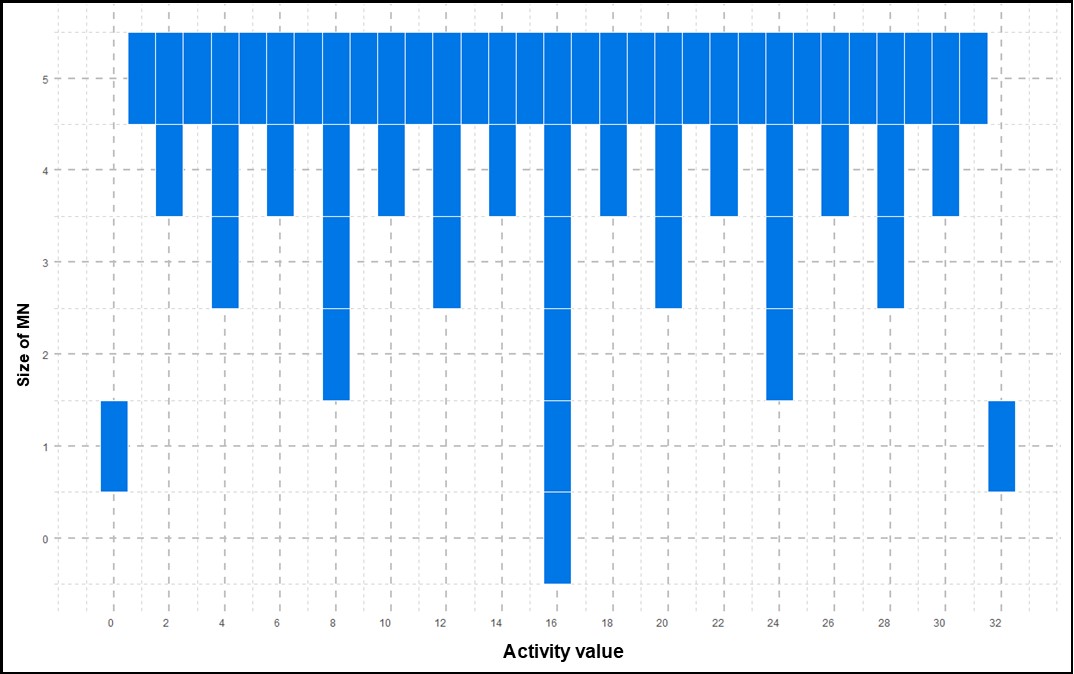}
\caption{Size of MN vs. Activity value for local functions with $\vert A \vert = 2$ and $\vert S \vert = 5$.}
\label{fig1}
\end{figure}

\begin{figure}
\centering
\includegraphics[scale=.7]{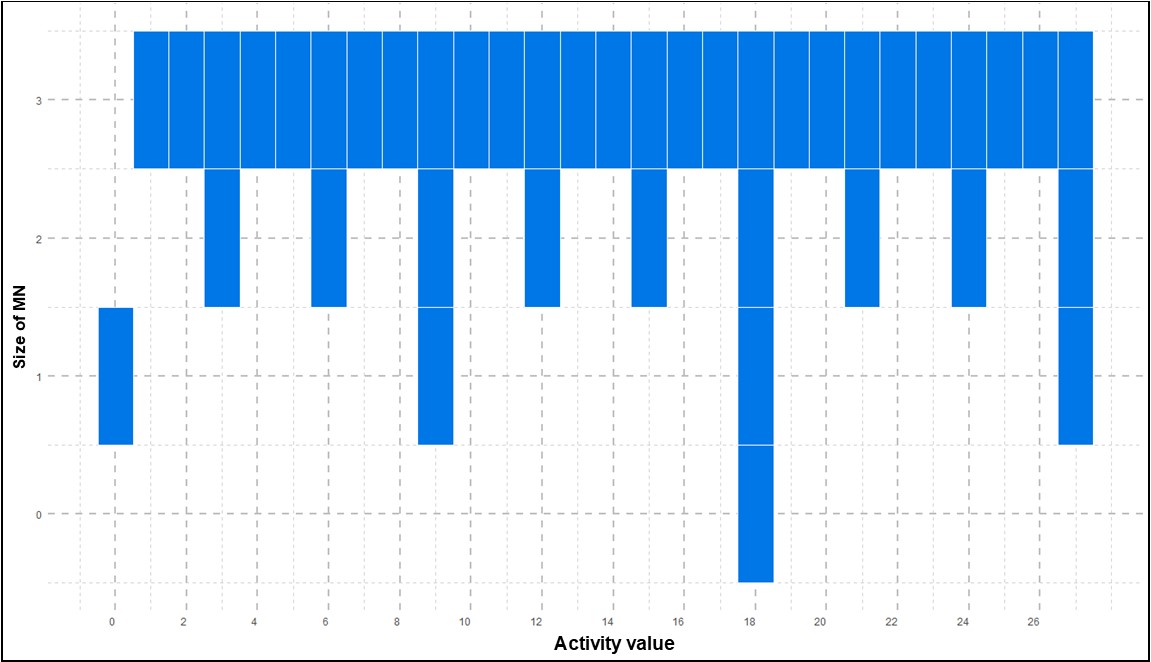}
\caption{Size of MN vs. Activity value for local functions with $\vert A \vert = 3$ and $\vert S \vert = 3$.}
\label{fig2}
\end{figure}

The goal of this paper is to mathematically explain the behaviors of the graphs given in Figures \ref{fig1} and \ref{fig2}. It turns out that these behaviors are independent of the universe $G$, but only depend on the set of the alphabet $A$ and the neighborhood $S$ of the local map $\mu : A^S \to A$. Our main result is the following theorem. 

\begin{theorem}\label{th-1}
Let $G$ be a group and let $A$ be a finite set such that $\vert A \vert \geq 2$. Let $S \subseteq G$ be a finite subset such that $e \in S$ and $\vert S \vert \geq 2$. 
\begin{enumerate}
\item Let $\mu : A^S \to A$ be a local map with minimal neighborhood $S_0$. If $e \in S_0$, then the activity value $\alpha(\mu)$ is a multiple of $\vert A \vert^{|S \setminus S_0|}$. Otherwise, if $e \not \in S_0$, then
\[\alpha(\mu) = \vert A \vert^{\vert S \vert} - \vert A \vert^{\vert S \vert-1}. \]  
\item For all $0 < k < \vert A \vert^{\vert S \vert}$, there exists a local function $\mu : A^S \to A$ whose minimal neighborhood is $S$ and $\alpha(\mu) = k$.
\item If $\vert A \vert \geq 3$, there exists a local function $\mu : A^S \to A$ whose minimal neighborhood is $S$ and $\alpha(\mu) = \vert A \vert^{\vert S \vert}$. 
\end{enumerate}
\end{theorem}

This paper is an extend version of the paper \cite{CV24} which was presented at AUTOMATA 2024 at Durham University, United Kingdom. Besides various changes in terminology and notation in order to be more accessible to a wider community (e.g., in \cite{CV24} we used the term \emph{minimal memory set} instead of minimal neighborhood, and the term \emph{generating pattern} instead of active transition), we extended all the main results presented in \cite{CV24}. Specifically, in \cite{CV24} were only able to explain some aspects of the behavior of Figures \ref{fig1} and \ref{fig2} (like the fact when $\alpha(\mu)$ is not a multiple of $\vert A \vert$, then the minimal neighborhood of $\mu : A^S \to A$ must be $S$ itself), while the above Theorem \ref{th-1} gives a complete explanation of this behavior.   

The structure of this paper is as follows. In Section 2, we prove some basic facts about the minimal neighborhood and the activity value of a local rule, and we show that these two notions behave well under symmetries induced by automorphisms of the universe $G$ or the alphabet $A$. In Section 3, we present results on the links between the size of the minimal neighborhood and the activity value of a local map $\mu : A^S \to A$, including the proof of Theorem 1. Finally, in Section 4, we discuss some possibilities for future work.

\section{Basic results}

\subsection{Minimal neighborhood}

We assume that the alphabet $A$ is a finite set with at least two elements and that $\{0,1\} \subseteq A$. In this paper, we consider the \emph{shift action} of $G$ on $A^G$ as the function $\cdot : G \times A^G \to A^G$ defined by 
\[ (g \cdot x)(h) := x(hg), \quad \forall x \in A^G, g, h \in G. \]
This indeed satisfies the axioms of a group action, namely $e \cdot x = x$ for all $x \in A^G$, and $g \cdot (h \cdot x) = (gh) \cdot x$, for all $g,h \in G$, $x \in A^G$. However, when $G$ is a nonabelian group, it is important to write $x(hg)$, and not $x(gh)$, in order to be able to verify these axioms. Alternatively, the shift action is sometimes defined (see \cite{CAandG}) as the function $\star :G \times A^G \to A^G$ defined by
\[ (g \star x)(h) := x(g^{-1}h), \quad \forall x \in A^G, g, h \in G. \] 
These two actions are in fact \emph{equivalent group actions} in the sense that there exists a bijection $\phi : A^G \to A^G$, given by $\phi(x)(h) := x(h^{-1})$ such that
\[ \phi( g \cdot x) = g \star \phi(x), \quad \forall g \in G, x \in A^G.  \]

When the universe is $G = \mathbb{Z}$, the configuration space $A^\mathbb{Z}$ may be identified with the set of bi-infinite sequences 
\[ x = \dots x_{-2} x_{-1} x_{0} x_{1} x_{2} \dots \]
for all $x \in A^\mathbb{Z}$, where $x_{k} := x(k) \in A$. The shift action of $\mathbb{Z}$ on $A^\mathbb{Z}$ is equivalent to left and right shifts of the bi-infinite sequences. 

Let $\tau : A^G \to A^G$ be a cellular automaton according to Definition \ref{def-ca}. We say that $\tau$ \emph{admits} a neighborhood $S \subseteq G$ if there exists a local map $\mu : A^S \to A$ that defines $\tau$; this means that  
\[ \tau(x)(g) = \mu( (g \cdot x) \vert_S), \quad \forall x \in A^G, g \in G. \]

\begin{example}\label{Elementary}
Let $G := \mathbb{Z}$ and $S := \{-1,0,1\} \subseteq G$. The relationship between a cellular automaton $\tau : A^G \to A^G$ with local defining map $\mu : A^S \to A$ is described as follows:
\[ \tau( \dots x_{-1} x_{0} x_{1} \dots ) = \dots \mu(x_{-2},x_{-1},x_{0}) \mu(x_{-1},x_0,x_1) \mu(x_{0}, x_1, x_2) \dots \]
In this setting, it is common to define a local map $\mu : A^S \to A$ via a table that enlists all the elements of $z \in A^S$, which are identified with tuples in $z_{-1}z_0 z_1 \in A^3$; explicitly
\[ \begin{tabular}{c|cccccccc}
$z\in A^S$ & $111$ & $110$ & $101$ & $100$ & $011$ & $010$ & $001$ & $000$ \\ \hline
$\mu(z) \in A$ & $a_1$ & $a_2$ & $a_3$ & $a_4$ & $a_5$ & $a_6$ & $a_7$ & $a_8$
\end{tabular}\] 
where $a_i \in A$. When $A=\{0,1\}$, cellular automata that admit a neighborhood $S = \{ -1,0,1\} \subseteq \mathbb{Z}$ are known as \emph{elementary cellular automata} (ECA) \cite[Sec. 2.5]{Kari}, and they are labeled with a \emph{Wolfram number}, which is the decimal number corresponding to the binary number $a_1a_2 \dots a_8$. 
\end{example}

We write $\mu \sim \nu$ if the local maps $\mu : A^S \to A$ and $\nu : A^T \to A$ define the same cellular automaton. This defines an equivalence relation, and it holds that 
\[ \mu \sim \nu \quad \Leftrightarrow \quad \mu( x \vert_S) = \nu( x \vert_T), \ \forall x \in A^G.  \]

\begin{definition}
The \emph{minimal neighborhood} of a cellular automaton $\tau : A^G \to A^G$, denoted by $\MN(\tau)$ is a neighborhood admitted by $\tau$ of smallest cardinality. The minimal neighborhood of a local map $\mu : A^S \to A$, denoted by $\MN(\mu)$, is the minimal neighborhood of the cellular automaton defined by $\mu$.  
\end{definition}

The following result is Proposition 1.5.2 in \cite{CAandG}.

\begin{proposition}\label{prop-mn}
Let $\tau : A^G \to A^G$ be a cellular automaton.
\begin{enumerate}
\item $\tau$ admits neighborhood $S \subseteq G$ if and only if $\MN(\tau) \subseteq S$ and $S$ is finite.
\item The minimal nieghborhood of $\tau$ is unique. 
\end{enumerate}
\end{proposition}

For $s \in S$ and $x \in A^S$, we write 
\[ [x]_s := \{ y \in A^S : x \vert_{S \setminus \{s\}} = y \vert_{S \setminus \{ s \}} \}.  \] 

The following is an easy consequence of the definition. 

\begin{lemma}
For any $s \in S$, the set $\{ [x]_s : x \in A^S \}$ is a uniform partition of $A^S$, with $\vert [x]_s \vert = \vert A \vert$, for all $x \in A^S$.
\end{lemma}

\begin{definition}\label{essential}
We say that an element $s \in S$ is \emph{essential} for a local map $\mu : A^S \to A$ if there exist $z,w \in A^S$ such that $[z]_s = [w]_s$ but $\mu(z) \neq \mu(w)$. 
\end{definition}

\begin{example}\label{Rule110}
Consider the elementary cellular automaton $\mu : A^{S} \to A$, with $S = \{-1,0,1\}$, given by the following table:
\[ \begin{tabular}{c|cccccccc}
$z\in A^S$ & $111$ & $110$ & $101$ & $100$ & $011$ & $010$ & $001$ & $000$ \\ \hline
$\mu(z) \in A$ & $0$ & $1$ & $1$ & $0$ & $1$ & $1$ & $1$ & $0$
\end{tabular}\] 
This has Wolfram number 110. In this case, all the elements of $S$ are essential for $\mu$. For example, $s:=-1 \in S$ is essential for $\mu$ because we may take $z=111$ and $w = 011$. These patterns satisfy that $[z]_s = \{ z,w \} = [w]_s$, and $\mu(z) = 0 \neq 1 = \mu(w)$, which shows that $s = -1$ is essential for $\mu$. Intuitively, this tells us that we cannot cross out the $-1$ coordinate of the patterns in the table that defines $\mu$.    
\end{example}

\begin{proposition}[c.f. Exercise 1.24 in \cite{ExCA}]\label{le-mms}
Let $\mu : A^S \to A$ be a local map. Then, 
\[ \MN(\mu) = \{ s \in S : s \text{ is essential for } \mu \}. \]
\end{proposition}
\begin{proof}
Let $S_0 := \MN(\mu)$. By Proposition \ref{prop-mn}, we have $S_0 \subseteq S$. Suppose that $s \in S$ is not essential for $\mu$. Define $\mu^\prime : A^{S \setminus \{s\}} \to A$ by $\mu^\prime(y) := \mu(\hat{y})$, for all $y \in A^{S \setminus \{s\}}$, where $\hat{y} \in A^S$ is any extension of $y$. The function $\mu^\prime$ is well-defined because $s$ is not essential for $\mu$, so for all $z,w \in A^S$ with $[z]_s = [w]_s$ we have $\mu(z) = \mu(w)$. Moreover, $\mu \sim \mu^\prime$, so again by Proposition \ref{prop-mn}, we have $S_0 \subseteq S \setminus \{s\}$. Hence, $s \not \in S_0$. 

Conversely, suppose there is $s \in S \setminus S_0$. Let $\mu_0 : A^{S_0} \to A$ be the local map associated with $S_0$ such that $\mu_0 \sim \mu$. If $s$ is essential for $\mu$, there exist $z,w \in A^S$ such that $[z]_s = [w]_s$ but $\mu(z) \neq \mu(w)$. However, $z \vert_{S \setminus \{s\}} = w \vert_{S \setminus \{s\}}$ implies that $\mu_0( z \vert_{S_0}) = \mu_0(w \vert_{S_0})$, as $s \not \in S_0$. This contradicts that $\mu \sim \mu_0$. Therefore, $s$ is not essential for $\mu$.   
\end{proof}

The previous proposition give us a practical algorithm to find the minimal neighborhood of a local map $\mu : A^S \to A$ with time complexity $q n \log_q(n)$, where $q= \vert A \vert$, since the input $\mu$ may be seen as a string of $n = q^{\vert S \vert}$ symbols from $A$.  \bigskip

\begin{algorithm}[H] \label{algo}
\SetAlgoLined
\caption{Minimal neighborhood of a local map.}

\KwIn{Local map $\mu : A^S \to A$}
\KwOut{Minimal neighborhood $\MN(\mu)$}
$\MN(\mu) := \emptyset$\;
\For{$s \in S$}{
   \For{$z \in A^S$}{
   \For{$w \in [z]_s, \ w \neq z$,}{
    \If{$(\mu(z) \neq \mu(w)) $}{$\MN(\mu) := \MN(\mu) \cup \{s\}$\;  \textbf{break}}
  }}}  
\Return $\MN(\mu)$ \medskip
\end{algorithm} \bigskip


\subsection{Activity value}

For the rest of the paper, assume that $S$ is a finite subset of $G$ such that $e \in S$. 

\begin{definition}\label{active}
The set of \emph{active transitions} of a local map $\mu : A^S \to A$, with $e \in S$, is defined by 
\[ \mathcal{T} := \{ z \in A^S : \mu(z) \neq z(e) \},  \]
The \emph{activity value} of $\mu : A^S \to A$ is $\alpha(\mu) := \vert \mathcal{T} \vert$. The set of \emph{passive transitions} of $\mu$ is defined by
\[ \mathcal{P} := A^S \setminus \mathcal{T}. \]   
\end{definition}

It is clear that the number of passive transitions $\mathcal{P}$ also determines the activity value of $\mu : A^S \to A$ because
\[ \alpha(\mu) = \vert A^S \setminus \mathcal{P} \vert.  \]

\begin{example}\label{Rule110}
Consider the elementary cellular automaton $\mu : A^{S} \to A$, with $S = \{-1,0,1\}$, given as in Example \ref{Rule110}. We see from the table that defines $\mu$ the set of active transitions is $\mathcal{T} = \{111, 101,001\}$, so $\alpha(\mu) = 3$, and the set of passive transitions is $\mathcal{P} = \{ 110, 100, 011, 010, 000 \}$.
\end{example}

\begin{lemma}
Let $\mu : A^S \to A$ be a local map with active transitions $\mathcal{T}$ and let $\tau : A^G \to A^G$ be the cellular automaton defined by $\mu$. 
\begin{enumerate}
\item $\tau$ is equal to the identity function if and only if $\mathcal{T} = \emptyset$. 
\item The set of fixed points $\Fix(\tau) := \{ x \in A^G : \tau(x) = x \}$ is equal to the subshift $X_\mathcal{T} \subseteq A^G$ with forbidden patterns $\mathcal{T}$, which is defined by
\[ X_\mathcal{T} := \{ x \in A^G : (g \cdot x)\vert_S \notin \mathcal{T}, \forall g \in G \}.   \] 
\end{enumerate}
\end{lemma}
\begin{proof}
Part (1) follows from the definition. For part (2), let $x \in X_{\mathcal{T}}$. Then, $(g \cdot x)\vert_S \not\in \mathcal{T}$, for all $g \in G$, so it follows that $(g \cdot x)\vert_S$ is a passive transition for $\mu$:
\[ \tau(x)(g) = \mu( (g \cdot x) \vert_S) = (g \cdot x)(e) = x(g), \quad \forall g \in G.   \]
Therefore, $x \in \Fix(\tau)$. Conversely, suppose that $x \not \in  X_{\mathcal{T}}$, so there exists $g \in G$ such that $(g \cdot x)\vert_S \in \mathcal{T}$. Then,
\[ \tau(x)(g) = \mu( (g \cdot x) \vert_S)  \neq ( g \cdot x)(e) = x(g).  \]
This shows that $\tau(x) \neq x$, so $x \not \in \Fix(\tau)$.
\end{proof}


\subsection{Symmetries}\label{sec-sym}

In \cite[Prop. 1]{Pedro1}, it was shown that the activity value preserved under \emph{symmetries} of elementary cellular automata. These so-called \emph{symmetries} are defined by the well-known automorphisms of one-dimensional automata of \emph{reflection} and \emph{complementation} (see \cite{CRG20,Wolfram}). In this section, we present the generalized versions of these symmetries when the universe is an arbitrary group $G$, and we show that both the size of the minimal neighborhood and the the activity value are always preserved under these symmetries. 

For each bijection of the alphabet $f : A \to A$ we may define an invertible cellular automaton $f_* : A^G \to A^G$ by 
 \[ f_*(x) := f \circ x, \quad \forall x \in A^G. \]
 Conjugation by $f_*$ induces an automorphism of the monoid of all cellular automata; explicitly, for any cellular automaton $\tau : A^G \to A^G$ we define a cellular automaton $\tau^f : A^G \to A^G$ by
 \[  \tau^f : = (f_*)^{-1} \circ \tau \circ f_*. \] 
When $A=\{0,1\}$ and $f : A \to A$ is the transposition $f(0)=1$ and $f(1)=0$, the cellular automaton $\tau^f$ is precisely the \emph{complement} of $\tau$. 

For each group automorphism $\phi : G \to G$, we may define a homeomorphism (i.e. continuous bijection with continuous inverse with respect to the prodiscrete topology) $\phi^* : A^G \to A^G$ by
\[ \phi^*(x) := x \circ \phi, \quad \forall x \in A^G. \]
The function $\phi^*$ is not a cellular automaton because it is not $G$-equivariant, but it is a \emph{$\phi$-cellular automaton} as introduced in \cite{GCA}. As shown in \cite[Sec. 4]{GCA}, conjugation by $\phi^*$ induces an automorphism of the monoid of all cellular automata; explicitly, for any cellular automaton $\tau : A^G \to A^G$ we define a cellular automaton $\tau^\phi : A^G \to A^G$ by
 \[  \tau^\phi : = (\phi^*)^{-1} \circ \tau \circ \phi^*. \] 
For the case when $G=\mathbb{Z}$ and $\phi : \mathbb{Z} \to \mathbb{Z}$ is the only non-trivial group automorphism given by $\phi(k) := -k$, for all $k \in \mathbb{Z}$, then the cellular automaton $\tau^\phi$ is known as the \emph{reflection} of $\tau$.  

\begin{lemma}
Let $\tau : A^G \to A^G$ be a cellular automaton, let $f : A \to A$ be a bijection, and let $\phi : G \to G$ be an automorphism. 
\begin{enumerate}
\item If $\mu : A^S \to A$ is a local function that defines $\tau$, then a local function $\mu^f : A^S \to A$ for $\tau^f$ is given by
\[ \mu^f (z) = f^{-1} \circ \mu ( f \circ z), \quad \forall z \in A^S.  \]

\item  If $\mu : A^S \to A$ is a local function that defines $\tau$, then a local function $\mu^\phi : A^{\phi(S)} \to A$ of $\tau^f$ is given by
\[ \mu^\phi (w) = \mu( w \circ \phi \vert_S ), \quad \forall w \in A^{\phi(S)}.  \]
\end{enumerate}
\end{lemma}
\begin{proof}
\begin{enumerate}
\item For all $x \in A^G$ we have
\[ \tau^f(x)(e) = f^{-1} \circ \tau(f \circ x)(e) = f^{-1}(\mu( (f \circ x)\vert_S)).  \]
This shows that the local map $\mu^f : A^S \to A$ defines $\tau^f$. 

\item For all $x \in A^G$ we have
\[ \tau^\phi(x)(e) = \tau(x  \circ \phi) \circ \phi^{-1}(e) = \tau(x \circ \phi)(e) = \mu( (x \circ \phi)\vert_S) =  \mu( x \vert_{\phi(S)} \circ \phi\vert_S ) ,  \]
where we have used the fact that $\phi^{-1}(e)=e$, because $\phi^{-1}$ is an automorphism of $G$. This shows that the local map $\mu^\phi : A^{\phi(S)} \to A$ defines $\tau^f$.
\end{enumerate}
\end{proof}

\begin{proposition}\label{sym-ac}
Let $S$ be a finite subset of $G$ such that $e \in S$. For any local function $\mu : A^S \to A$, any bijection $f : A \to A$ and any automorphism $\phi : G \to G$, we have
\[ \alpha(\mu) = \alpha(\mu^f) = \alpha(\mu^\phi). \]
\end{proposition}
\begin{proof}
Let $\mathcal{P}$, $\mathcal{P}^f$, and $\mathcal{P}^\phi$ be the sets of passive transitions of $\mu$, $\mu^f$, and $\mu^\phi$, respectively. We will show that $f_* : \mathcal{P}^f \to \mathcal{P}^f$ given by $f_*(w) = f \circ w$ for all $w \in \mathcal{P}^f$ is a well-defined bijection. If $w \in \mathcal{P}^f$, then
\[ w(e) =  \mu^f (w) =  f^{-1} \circ \mu ( f \circ w)  \quad \Longrightarrow \quad (f \circ w)(e) = \mu(f \circ w), \] 
which shows that $f \circ w \in \mathcal{P}$. On the other hand, if $z \in \mathcal{P}$, we will show that $f^{-1} \circ z \in \mathcal{P}^f$. Indeed,
\[ \mu^f(f^{-1} \circ z) =  f^{-1} \circ \mu ( f \circ f^{-1} \circ z) = f^{-1} \circ \mu(z) = (f^{-1} \circ z)(e).    \]
This proves that $\vert \mathcal{P} \vert = \vert \mathcal{P}^f \vert$, so $\alpha(\mu) = \alpha(\mu^f)$.  

Now, we will show that $(\phi \vert_S)^* :  \mathcal{P}^\phi \to \mathcal{P}$ given by $(\phi \vert_S)^*(w) = w \circ \phi \vert_S$ for all $w \in A^{\phi(S)}$ is a well-defined bijection. If $w \in \mathcal{P}^\phi$, then
\[  (w \circ \phi \vert_S)(e) =w(e) = \mu^\phi(w) = \mu( w \circ \phi \vert_S ),  \]  
which shows that $w \circ \phi \vert_S \in \mathcal{P}$. On the other hand, we will show that if $z \in \mathcal{P}$, then $z \circ \phi^{-1}\vert_{\phi(S)} \in \mathcal{P}^\phi$. Indeed,
\[ \mu^\phi(z \circ \phi^{-1}\vert_{\phi(S)}) = \mu( z \circ \phi^{-1}\vert_{\phi(S)} \circ \phi \vert_S ) = \mu(z) = z (e) = (z \circ \phi^{-1}\vert_{\phi(S)})(e).    \]
This proves that $\vert \mathcal{P} \vert = \vert \mathcal{P}^\phi \vert$, so $\alpha(\mu) = \alpha(\mu^\phi)$.  
\end{proof}

We finish this section by showing that the size of the minimal neighborhood is also preserved under symmetries. 

\begin{proposition}
For any local function $\mu : A^S \to A$, any bijection $f : A \to A$ and any automorphism $\phi : G \to G$, we have
\[ \MN(\mu) = \MN(\mu^f) \text{ and } \MN(\mu^\phi) = \phi(\MN(\mu)). \]
In particular,
\[ \vert \MN(\mu) \vert = \vert \MN(\mu^f) \vert = \vert \MN(\mu^\phi) \vert.  \]
\end{proposition}
\begin{proof}
Observe first that for any $z,w \in A^S$, we have $\mu(z) = \mu(w)$ if and only if $\mu^f( f^{-1} \circ z) = \mu^f( f^{-1} \circ w)$, and, for any $s \in S$, $[z]_s = [w]_s$ if and only if $[f \circ z]_s = [f \circ w]_s$. This shows that $s \in S$ is essential for $\mu$ if and only if $s$ is essential for $\mu^f$, and the first equality follows. 

For the second equality, note that for any $z,w \in A^S$, we have $\mu(z) = \mu(w)$ if and only if $\mu^\phi( z   \circ \phi^{-1}\vert_{\phi(S)} ) = \mu^\phi( w  \circ \phi^{-1}\vert_{\phi(S)})$, and, for any $s \in S$, $[z]_s = [w]_s$ if and only if $[z  \circ \phi^{-1}\vert_{\phi(S)} ]_{\phi(s)} = [w \circ \phi^{-1}\vert_{\phi(S)}]_{\phi(s)}$. This shows that $s \in S$ is essential for $\mu$ if and only if $\phi(s) \in \phi(S)$ is essential for $\mu^\phi$, and the result follows.   
\end{proof}


\section{Connecting the minimal neighborhood and the activity value} 

Our first result connects the activity value of a local map $\mu : A^S \to A$ with the activity value of the ``reduced'' local map with respect to the minimal neighborhood of $\mu$. 

\begin{lemma}\label{le-div}
Let $\mu : A^S \to A$ be a local map and let $S_0 := \MN(\mu)$. Let $\mu_0 : A^{S_0} \to A$ be the local map associated to $S_0$ such that $\mu \sim \mu_0$. If $e \in S_0$, then 
\[ \alpha(\mu) = \alpha(\mu_0) \vert A \vert^{|S \setminus S_0|}.  \]
In particular, $ \vert A \vert^{|S \setminus S_0|}$ divides $\alpha(\mu)$. 
\end{lemma}
\begin{proof}
Let $\mathcal{T} \subseteq A^S$ and $\mathcal{T}_0 \subseteq A^{S_0}$ be the sets of active transitions of $\mu$ and $\mu_0$, respectively. Let $f : A^S \to A^{S_0}$ be the restriction function defined by $f(z) := z \vert_{S_0}$, for all $z \in A^S$. This is well-defined because $S_0 \subseteq S$ by Proposition \ref{prop-mn}. Since $\mu \sim \mu_0$, then $f(z) \in \mathcal{T}_0$ if and only if $z \in \mathcal{T}$, so we may consider the restriction $f \vert_{\mathcal{T}} : \mathcal{T} \to \mathcal{T}_0$. Observe that each $w \in \mathcal{T}_0$ has precisely $\vert A \vert^{|S \setminus S_0|}$ preimages under $f \vert_{\mathcal{T}}$. As the preimage sets form a partition of $\mathcal{T}$, we obtain that
\[ \vert \mathcal{T} \vert = \vert \mathcal{T}_0 \vert \vert A \vert^{|S \setminus S_0|}.   \]
The result follows.
\end{proof}

\begin{lemma}\label{le-3}
Let $\mu : A^S \to A$ be a local map with active transitions $\mathcal{T}$ and passive transitions $\mathcal{P}$.
\begin{enumerate}
 \item If there exist $z, w \in \mathcal{P}$, such that $z \neq w$ and $[z]_e = [w]_e$, then $e \in S$ is essential for $\mu$. 
\item For $s \in S \setminus \{e\}$, if there exist $z \in \mathcal{T}$ and $w \in \mathcal{P}$ such that $[z]_s = [w]_s$, then $s$ is essential for $\mu$.
 \end{enumerate}
\end{lemma}
\begin{proof}
For part (1), observe that $z(e) \neq w(e)$ because $z \neq w$ and $[z]_e = [w]_e$. Since $z, w \in \mathcal{P}$, we have
\[ \mu(z) = z(e) \neq w(e) = \mu(w).  \]
It follows that $e$ is essential for $\mu$. 

For point (2), observe that $z(e) = w(e)$ because $[z]_s = [w]_s$ and $s \neq e$. Then,
\[ \mu(z) = z(e) = w(e) \neq \mu(w).   \]
It follows that $s$ is essential for $\mu$.  
\end{proof}

\begin{corollary}\label{cor-1}
Let $\mu : A^S \to A$ be a local map with active transitions $\mathcal{T}$ and passive transitions $\mathcal{P}$.
\begin{enumerate}
\item If $e \in S$ is not essential for $\mu$, then $[z]_e \cap \mathcal{P} = \{ z \}$, for all $z \in \mathcal{P}$.
\item If $s \in S \setminus \{e\}$ is not essential for $\mu$, then $[z]_s \subseteq \mathcal{T}$, for all $z \in \mathcal{T}$.
\end{enumerate}
\end{corollary}
\begin{proof}
This is just the contrapositive of Lemma \ref{le-3}.
\end{proof}

\begin{lemma}\label{le-e}
Let $\mu : A^S \to A$ be a local map with $e \in S$ and with $\vert S \vert \geq 2$. If $e$ is not essential for $\mu$, then 
\[ \alpha(\mu) = \vert A \vert^{\vert S \vert} - \vert A \vert^{\vert S \vert-1}.  \]
\end{lemma}
\begin{proof}
Since $\vert S \vert \geq 2$, we have $A^{S \setminus \{ e\}} \neq \emptyset$. Let $\mathcal{P}$ be the set of passive transitions of $\mu$. We shall prove that the function $\psi : \mathcal{P} \to  A^{S \setminus \{ e\}}$, defined by 
\[ \psi(z) := z \vert_{S \setminus \{e\}}, \quad \forall z \in \mathcal{P}, \]
is a bijection.
\begin{itemize}
\item We will show that $\psi$ is injective. if $\psi(z) = \psi(w)$ for some $z,w \in \mathcal{P}$, then $[z]_e = [w]_e$. If follows by Corollary \ref{cor-1} (1) that
\[  [z]_e \cap \mathcal{P} = \{z\} \quad \text{ and } \quad  [w]_e \cap \mathcal{P} = \{w\} . \]
Therefore, $z=w$. It follows that $\psi$ is injective. 

\item We will show that $\psi$ is surjective. We claim that for every $y \in A^{S \setminus \{e\}}$ there exists $\hat{y} \in \mathcal{P}$ such that $\hat{y}\vert_{S \setminus \{e\}} = y$. First observe that it follows from definition that if $e$ is not essential for $\mu$, then $\mu$ is constant on $[x]_e$ for all $x \in A^S$. Let $w \in A^S$ be any extension of $y$ and suppose that $\mu(w^\prime) = a$ for all $w^\prime \in [w]_e$. Choose $\hat{y} \in [w]_e$ such that $\hat{y}(e) = a$. Then, $\hat{y}$ is a passive transition for $\mu$, so $\hat{y} \in \mathcal{P}$. Clearly, $\hat{y} \vert_{S \setminus \{e\}} = w\vert_{S \setminus \{e\}} = y$, so the function $\psi$ is surjective.      
\end{itemize}     

Since $\psi$ is a bijection, then $\vert \mathcal{P} \vert = \vert A \vert^{\vert S \vert-1}$, which is equivalent to 
\[ \alpha(\mu) = \vert A^S \setminus \mathcal{P} \vert = \vert A \vert^{\vert S \vert} - \vert A \vert^{\vert S \vert-1}. \] 
\end{proof}

\begin{remark}
The case when the activity value is $\alpha(\mu) = \vert A \vert^{\vert S \vert} - \vert A \vert^{\vert S \vert-1}$ seems to be special for other reasons, besides the case when $e$ is not essential for $\mu$. For example, if $A$ is a finite field and $\mu : A^S \to A$ is a local linear map different from the projection to $e$, then we must also have $\alpha(\mu) = \vert A \vert^{\vert S \vert} - \vert A \vert^{\vert S \vert-1}$. This may be been seen as the set of passive transitions $\mathcal{P}$ of $\mu$ is the kernel of the nonzero linear map $\mu - \pi_e$, where $\pi_e : A^S \to A$ is the projection to $e$ defined by $\pi_e(x) = x(e)$, for all $x \in A^G$, so by the Rank-Nullity theorem we deduce that
\[ \dim(\mathcal{P}) = \dim(A^S) - 1 = \vert S \vert - 1.   \]
This shows that $\vert \mathcal{P} \vert = \vert A \vert^{\vert S \vert-1}$, so $\alpha(\mu) = \vert A \vert^{\vert S \vert} - \vert A \vert^{\vert S \vert-1}$.
\end{remark}

\begin{corollary}\label{not-multiple}
Let $\mu : A^S \to A$ be a local map with $\vert S \vert \geq 2$. If $\vert A \vert \nmid \alpha(\mu)$, then $\MN(\mu) =  S$.
\end{corollary}
\begin{proof}
If $e$ is not essential for $\mu$, then Lemma \ref{le-e} shows that $\vert A \vert$ divides $\alpha(\mu)$. Hence, $e$ must be essential for $\mu$. By Lemma \ref{le-div}, $ \vert A \vert^{|S \setminus S_0|}$ divides $\alpha(\mu)$, so we must have $|S \setminus \MN(\mu)|=0$. It follows that $\MN(\mu) =  S$.
\end{proof}

The previous corollary already allows us to improve Algorithm \ref{algo} for $\vert S \vert \geq 2$: we may begin by counting the activity value of $\alpha(\mu)$, and if it is not a multiple of $\vert A \vert$, then $\MN(\mu) =  S$.

For $s \in S$, define the pattern $\delta_s \in A^S$ by $\delta_s(t) = \delta_{s,t}$, for all $t \in S$, where $\delta_{s,t}$ is the Kronecker delta function (i.e. $\delta_s(s)=1$ and $\delta_s(t) = 0$ for all $t \in S \setminus \{s\}$). For $a \in A$, let $a^S \in A^S$ be the constant pattern given by $a^S(s) = a$, for all $s \in S$.

\begin{theorem}\label{th-last}
Let $S \subseteq G$ be a finite subset such that $e \in S$ and $\vert S \vert \geq 2$. For all $0 < k < \vert A \vert^{\vert S \vert}$ there exists a local function $\mu : A^S \to A$ such that
\[ \MN(\mu) = S  \quad \text{ and } \quad \alpha(\mu) = k.  \]
Furthermore, if $\vert A \vert \geq 3$, there exists also a local function $\nu : A^S \to A$ such that
\[ \MN(\nu) = S  \quad \text{ and } \quad  \alpha(\mu) = \vert A \vert^{\vert S \vert}.  \]
\end{theorem}
\begin{proof}
For all the required values of $k$, we will choose a local function $\mu : A^S \to A$ with $k:= \alpha(\mu)$, active transitions $\mathcal{T}$, and passive transitions $\mathcal{P}$, satisfying that $\MN(\mu) = S$  

First assume that $k \neq \vert A \vert^{\vert S \vert}$ and $k \neq \vert A \vert^{\vert S \vert} - \vert A \vert^{\vert S \vert-1}$. By Lemma \ref{le-e}, any choice of $\mu : A^S \to A$ with activity value $k$ satisfies that $e$ is essential for $\mu$. If $k \geq \vert S \vert -1$, we may choose $\mu$ with $\mathcal{T}$ and $\mathcal{P}$ such that 
\[ \{ \delta_s : s \in S \setminus \{e\} \} \subseteq \mathcal{T} \text{ and } 0^S \in \mathcal{P}. \]
By Lemma \ref{le-3} (2), every element of $S \setminus \{e\}$ is essential for $\mu$, so the minimal neighborhood of $\mu$ is $S$.

If $k < \vert S \vert -1$, then
\[ \vert\mathcal{P}\vert = \vert A \vert^{\vert S \vert} - k \geq \vert A \vert^{\vert S \vert } - \vert S \vert +1. \]
In particular, $\vert \mathcal{P} \vert \geq \vert S \vert - 1$, so we may choose $\mu$ with $\mathcal{T}$ and $\mathcal{P}$ such that 
\[ \{ \delta_s : s \in S \setminus \{e\} \} \subseteq \mathcal{P} \text{ and } 0^S \in \mathcal{T}. \]
By Lemma \ref{le-3} (2), every element of $S \setminus \{e\}$ is essential for $\mu$, so the minimal neighborhood of $\mu$ is $S$.

Suppose that $k = \vert A \vert^{\vert S \vert} - \vert A \vert^{\vert S \vert-1}$. Since $\vert S \vert \geq 2$, then $k \geq \vert S \vert - 1$ and $\vert \mathcal{P} \vert = \vert A \vert^{\vert S \vert -1} \geq \vert A \vert \geq 2$. Then, we may choose $\mu$ with $\mathcal{T}$ and $\mathcal{P}$ such that 
\[ \{ \delta_s : s \in S \setminus \{e\} \} \subseteq \mathcal{T} \text{ and } \{ 0^S, \delta_e \} \in \mathcal{P}. \]
By Lemma \ref{le-3} (2), every element of $S \setminus \{e\}$ is essential for $\mu$, any by Lemma \ref{le-3} (1), $e$ is also essential for $\mu$, so $\MN(\mu)=S$.

Finally, assume that $k =  \vert A \vert^{\vert S \vert}$ and $\vert A \vert \geq 3$, so we may suppose that $\{0,1,2 \} \subseteq A$. Choose $\mu : A^S \to A$ such that 
\[ \mu(0^S) = 1, \quad \text{ and } \quad \mu(\delta_s) = 2, \ \forall s \in S. \]
It follows by definition that $\MN(\mu)=S$, and the result follows.
\end{proof}

The following corollary finally explains Figures \ref{fig1} and \ref{fig2}.

\begin{corollary}\label{cor-last}
Let $S \subseteq G$ be a finite subset such that $e \in S$ and $\vert S \vert \geq 2$. For all $0 < r \leq \vert S\vert$ and for all $0 < k < \vert A \vert^{\vert S\vert}$ such that $\vert A \vert^{\vert S \vert - r} \mid k$ there exists a local function $\mu : A^S \to A$ such that
\[ \vert\MN(\mu)\vert = r  \quad \text{ and } \quad \alpha(\mu) = k.  \]
Furthermore, if $\vert A \vert \geq 3$, there exists also a local function $\nu : A^S \to A$ such that
\[ \vert\MN(\nu)\vert = r  \quad \text{ and } \quad  \alpha(\mu) = \vert A \vert^{\vert S \vert}.  \]
\end{corollary}
\begin{proof}
For all $0 < r \leq \vert S\vert$, we may choose a local function $\mu : A^S \to A$ such that $\vert\MN(\mu)\vert = r$ and $e \in S$ (just build a local function with $r$ essential elements from $S$, including $e \in S$). By Lemma \ref{le-div},
\[ \alpha(\mu) = \alpha(\mu_0) \vert A \vert^{|S| - r},  \]
where $\mu_0 : A^{S_0} \to A$ is the reduced local function such that $\mu_0 \sim \mu$ and $S_0 = \MN(\mu)$. The result follows by applying Theorem \ref{th-last} to the local function $\mu_0$. 
\end{proof}


\subsection{Binary alphabet}

When $A= \{0,1\}$, the connections between the sizes of minimal neighborhoods and activity values of local functions presents further symmetries. In this case, a local function $\mu : A^S \to A$ is completely defined by its set of active transitions $\mathcal{T}$, because, for all $z \in A^S$, we have
\[ \mu(z) = \begin{cases}
z(e)^c & \text{ if} z \in \mathcal{T} \\
z(e) & \text{ if } z \not\in \mathcal{T} 
\end{cases}\]
where $c$ denotes the complement of the elements in $A=\{0,1\}$.

\begin{lemma}\label{le-4}
Let $A=\{0,1\}$ and let $\mu : A^S \to A$ be a local map with active transitions $\mathcal{T}$ and passive transitions $\mathcal{P}$.
\begin{enumerate}
 \item $e \in S$ is essential for $\mu$ if and only if there exist $z, w \in \mathcal{T}$, or $z,w \in \mathcal{P}$, such that $z \neq w$ and $[z]_e = [w]_e$.  
\item $s \in S \setminus \{e\}$ is essential for $\mu$ if and only if there exist $z \in \mathcal{T}$ and $w \in \mathcal{P}$ such that $[z]_s = [w]_s$.
 \end{enumerate}
\end{lemma}
\begin{proof}
For part (1), suppose that $e$ is essential for $\mu$. By definition, there are $z,w \in A^S$ such that $[z]_e = [w]_e$ and $\mu(z) \neq \mu(w)$. Note that $z(e) \neq w(e)$, as otherwise $z=w$. Since $A=\{0,1\}$, we must have $z(e)^c = w(e)$. If $z \in \mathcal{T}$ and $w \in \mathcal{P}$, then 
\[ \mu(z) = z(e)^c = w(e) = \mu(w),  \]
which is a contradiction. Therefore, either $z, w \in \mathcal{T}$, or $z,w \in \mathcal{P}$. Suppose now that there exist $z, w \in \mathcal{T}$, or $z,w \in \mathcal{P}$, such that $z \neq w$ and $[z]_e = [w]_e$. If $z,w \in \mathcal{P}$, it follows by Lemma \ref{le-3}(2) that $e$ is essential for $\mu$, so assume that $z, w \in \mathcal{T}$. As before, $z(e) \neq w(e)$, so $z(e)^c \neq w(e)^c$. Hence,
\[ \mu(z) = z(e)^c \neq w(e)^c = \mu(w). \]
This shows that $e$ is essential for $\mu$.  

The converse implication of part (2) follows by Lemma \ref{le-3}(2). For the direct implication of part (2), suppose that $[z]_s \subseteq \mathcal{T}$ for all $z \in \mathcal{T}$. Take arbitrary $z_1, z_2 \in A^S$ such that $[z_1]_s = [z_2]_2$. Since $s \neq e$, then $z_1(e) = z_2(e)$. If $z_1 \in \mathcal{T}$, then $z_2 \in \mathcal{T}$ by assumption, so
\[ \mu(z_1) = z_1(e)^c = z_2(e)^c = \mu(z_2), \]
If $z_1 \in \mathcal{P}$, then $z_2 \in \mathcal{P}$ by assumption, so
\[ \mu(z_1) = z_1(e) = z_2(e) = \mu(z_2).  \] 
This shows that $s$ is not essential for $\mu$. 
\end{proof}

The following proposition explains some further symmetries that appear in the case of a binary alphabet.  

\begin{proposition}\label{le-complement}
Let $A=\{0,1\}$ and let $\mu : A^S \to A$ be a local map with active transitions $\mathcal{T}$. Denote by $\mu^\prime : A^S \to A$ the local map with active transitions $A^S \setminus \mathcal{T}$. Then, 
\[ \MN(\mu) = \MN(\mu^\prime). \] 
\end{proposition}
\begin{proof}
By Lemma \ref{le-4} (1), $e$ is essential for $\mu$ if and only if there are $z, w \in \mathcal{T}$, or $z,w \in A^S \setminus \mathcal{T}$, such that $z \neq w$ and $[z]_e = [w]_e$; but again by Lemma \ref{le-4} (1), this holds if and only if $e$ is essential for $\mu^\prime$

By Lemma \ref{le-4} (2), $s \in S \setminus \{e\}$ is essential for $\mu$ if and only if there are $z \in \mathcal{T}$ and $w \in A^S \setminus \mathcal{T}$ such that $[z]_s = [w]_s$; but again by Lemma \ref{le-4} (2), this holds if and only if $s$ is essential for $\mu^\prime$.
\end{proof}

\begin{remark}
The local map $\mu^\prime : A^S \to A$ defined in Proposition \ref{le-complement} is different from the complement $\mu^c : A^S \to A$ defined in Section \ref{sec-sym}. Recall that by Proposition \ref{sym-ac}, $\alpha(\mu)=\alpha(\mu^c)$, but clearly 
\[ \alpha(\mu^\prime) = 2^{|S|} - \alpha(\mu). \]
\end{remark}

\section{Future work}

In this article, for a finite subset $S \subseteq G$ with $e \in S$, we considered the activity value of a local function $\mu : A^S \to A$ as the number of patterns $z \in A^S$ such that $\mu(z) \neq z(e)$. Since a cellular automaton $\tau : A^G \to A^G$ may be defined by many local functions, it is natural to define the \emph{activity value of $\tau$} by $\alpha(\tau) := \alpha(\mu_0)$, where $\mu_0 : A^{S_0} \to A$ is the local function associated with the minimal neighborhood $S_0 := \MN(\tau)$. However, note that this definition is only possible when $e \in \MN(\tau)$. The following is a list of problems that may be interesting for future work. 

\begin{enumerate}
\item Let $\tau, \sigma : A^G \to A^G$ be two cellular automata such that $e \in \MN(\tau)$ and $e \in \MN(\sigma)$. When does it hold that $e \in \MN(\tau \circ \sigma)$? In case that $e \in \MN(\tau \circ \sigma)$, is there any relation between $\alpha(\tau \circ \sigma)$, $\alpha(\tau)$ and $\alpha(\sigma)$? 

\item Let $\tau : A^G \to A^G$ be an invertible cellular automaton such that $e \in \MN(\tau)$. When does it hold that $e \in \MN(\tau^{-1})$? In case that $e \in \MN(\tau^{-1})$, is there any relation between $\alpha(\tau^{-1})$ and $\alpha(\tau)$? 

\item Count the exact number of local functions with a given activity value and size of minimal neighborhood. 

\item Generalize the definition of activity value as follows: \emph{activity value at $s \in S$} of a local function $\mu : A^S \to A$ is the number of patterns $z \in A^S$ such that $\mu(z) \neq z(s)$. 
\end{enumerate}

\section*{Acknowledgments}

The second author was supported by CONAHCYT \emph{Becas nacionales para estudios de posgrado}, Government of Mexico. We sincerely thank Maximilien Gadouleau for organizing the workshop AUTOMATA 2024 at Durham University and for his invitation to submit this extended version. We also thank Barbora Hudcov\'a and Nazim Fat\`es for very interesting discussions during this workshop.


\end{document}